\documentclass[
  10pt,
  pra,
  showpacs,
  preprintnumbers,
  twocolumn]
 {revtex4-1}
\usepackage{inputenc}
\usepackage[T1]{fontenc}
\usepackage{lmodern}
\usepackage{amsmath,amsfonts,amssymb,amstext,amscd,amsthm,dsfont}
\usepackage[mathscr]{eucal}
\usepackage{mathtools}
\usepackage{calrsfs,enumerate}
\usepackage[normalem]{ulem}
\usepackage{graphicx}
\usepackage{makeidx}
\usepackage{bbold}
\usepackage{blindtext}
\usepackage{qcircuitMP}
\usepackage{subcaption}
\usepackage[colorinlistoftodos]{todonotes}

\newcommand\ket[1]{\left|#1\right>}

\newcommand\projj[2]{\left|#1\right>\!\left<#2\right|}
\newcommand\proj[1]{\left|#1\right>\!\left<#1\right|}

\DeclareMathOperator{\Tr}{\operatorname{Tr}}
\newcommand{\deff}{\vcentcolon=}

\newcommand{\vspan}{\operatorname{span}}
\newcommand{\rab}{\rho_{AB}}

\newcommand{\ra}{\rho_{A}}

\newcommand{\id}{\textup{id}}

\newcommand{\I}{\openone}
\newcommand{\SR}{\textup{SR}}
\newcommand{\OSR}{\textup{OSR}}
\newcommand{\OSD}{\textup{OSD}}

\newcommand{\reff}[1]{Eq.~\eqref{#1}}

\newcommand*{\cD}{\mathcal{D}}

\newcommand*{\cH}{\mathcal{H}}

\newcommand*{\cL}{\mathcal{L}}

\newtheorem{Theorem}{Theorem}

\newtheorem{Lemma}{Lemma}

\begin{document}
\title{Correlation-assisted process tomography}   
\author{Matteo Caiaffa, Marco Piani}
\affiliation{SUPA and Department of Physics, University of Strathclyde, Glasgow G4 0NG, UK}
\begin{abstract}
Standard quantum process tomography on a $d$-dimensional input
is performed by 
preparing several states of an input probe that then evolve under the action of the quantum channel corresponding to the progress. The final states of the probe are reconstructed by means of state tomography. An alternative is offered by ancilla-assisted process tomography: a single probe-ancilla state is used, and the correlations existing between probe and ancilla are exploited to fully reconstruct the information on the channel. In order for ancilla-assisted process tomography to be possible, the probe-ancilla input state does not need to be entangled, but still needs to have maximal operator Schmidt rank. Here we establish and analyze a framework for process tomography that interpolates between these two methods, aiming at exploiting any correlations that may exist between probe and ancilla to allow process tomography with as few input preparations as possible, when the probe-ancilla state may be operator-Schmidt-rank deficient. The main object of our investigation is the minimal number of initial local operations on the input probe for a given starting probe-ancilla state that are needed to allow process tomography. We prove that such a number scales inversely proportional to the operator Schmidt rank of the input probe-ancilla state for arbitrary local input processing. We also provide results for the case where this initial local processing is restricted to be a unitary rotation, in particular showing that the mentioned scaling is satisfied for pure input entangled states, and that in the case of mixed states there might be extreme cases where the fixed input probe-ancilla state provides information on a $\approx d^2/2$-dimensional subspace of the operator space and a single additional input unitary rotation allows process tomography. 
\end{abstract}
\pacs{---}
\maketitle

\noindent
\section{Introduction}
Quantum information science promises the realisation of impressive advances in the field of computation and communication, as well as the advent of other novel quantum technologies \cite{nielsen2010quantum}. The variety and potential usefulness of such applications have made quantum information one of the most prolific physics research areas of the last decades. One critical element in accomplishing said progress is the ability to completely and accurately characterize quantum physical processes, a task known as quantum process (or channel) tomography (QPT). In the extensive literature upon the subject, one recognises two major, extreme examples of QPT: standard quantum process tomography (SPT) and ancilla-assisted quantum process tomography (AAPT). In the general case of SPT, an unknown  quantum channel $\Lambda$ acting on a $d$-level system (also called a qudit) can be reconstructed through its action on an ensemble of linearly independent input states \cite{nielsen2010quantum,  chuang1997prescription, poyatos1997complete}. In particular, a probe is prepared in a fixed set of $d^2$ input states $\{\rho_i\}$, which form a basis for the space of qudit linear operators. Each of the $\rho_i$ states goes through the process $\Lambda$ to be characterized, and the outputs $\Lambda[\rho_i]$ are determined using quantum state tomography \cite{nielsen2010quantum,d2003quantum,bisio2009optimal} (see Figure~\ref{fig:SPT}).
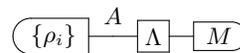
\begin{figure}[b]
\[
\Qcircuit @C=1em @R=.7em @! R {    
 & \prepareC{\{\rho_{i}\}}	& \ustick{A}\qw & \gate{\Lambda}& \measureD{M}
 }
\]
\caption{Standard quantum process tomography. To reconstruct the action of a channel $\Lambda$ acting on a $d$-dimensional system $A$, $d^2$ linearly independent input states $\{\rho_i\}$ are needed, with state tomography done on the outputs by means of measurement(s) $M$. Time goes from left to right. Single lines represent quantum systems, and boxes represent operations: a square box has quantum input and quantum output, while a D-shaped (reverse-D-shaped) box has only quantum input (quantum output).}
\label{fig:SPT}
\end{figure}
Once the outputs are known, the evolution under $\Lambda$ of any arbitrary operator can be determined uniquely by linearity , thus characterizing the channel. 
An alternative tomographic technique is offered by the renowned AAPT \cite{d2001quantum, altepeter2003ancilla, d2003imprinting} which, in contrast to SPT, needs only one single bipartite input state. That this is possible can be seen as a consequence of the correspondence between linear maps and linear operators established by the well-known  Choi-Jamio{\l}kowski isomorphism~\cite{choi1975completely,jamiolkowski1972linear}. In general, an ancillary system $B$ is prepared in a correlated state $\rho_{AB}$ with the quantum system subject to the channel to be determined, the probe $A$.
Complete information about the channel can be imprinted on the global state by the action of the process on the probe alone, and then extracted by state tomography on the bipartite output state (see Figure~\ref{fig:AAPT}).
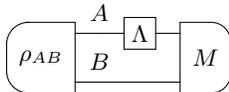
\begin{figure}
\[
\Qcircuit @C=1em @R=.7em @! R {    
 & \multiprepareC{1}{\rho_{AB}}	& \ustick{A}\qw			& \gate{\Lambda}	& \multimeasureD{1}{M}\\
 & \pureghost{\rho_{AB}}  		& \ustick{B}\qw			& \qw			& \ghost{M}		}	
\]
\caption{Ancilla-assisted quantum process tomography. One input state $\rho_{AB}$ suffices, as long as it has Operator Schmidt Rank equal to the square of the input dimension of the channel.}
\label{fig:AAPT}
\end{figure}

An input enabling enabling AAPT is the maximally entangled state $\ket{\Phi}_{AB}=\sum_{i=1}^dd^{-1/2}\ket{i}\otimes\ket{i}$, with the output $\rho_{\Lambda}= (\Lambda\otimes\id)[\proj{\Phi}_{AB}]$ simply being the Choi-Jamio{\l}kowski state isomorphic to $\Lambda$ \cite{choi1975completely,jamiolkowski1972linear,d2001quantum}. 
However, it was observed \cite{altepeter2003ancilla, d2003imprinting} that the key property for a bipartite input to enable AAPT is that of having maximal Operator Schmidt Rank (OSR; to be defined later), with a refining of this observation being that the channel discrimination power of a bipartite state is dictated by its smallest operator Schmidt coefficient~\cite{caiaffa2018}. It follows that, in principle, also non-entangled but correlated states can be used to perform AAPT. Bipartite states carrying a complete imprinting of a channel acting on one of the two subsystems were defined as \emph{faithful} in Ref.~\cite{d2003imprinting}. Nonetheless, non-faithful states can still be used to obtain substantial albeit partial information on the action of a channel. This observation suggests that the property of being faithful can be associated with a \emph{set} of bipartite states, the latter being faithful when any unknown channel can be fully retrieved from the tomographic reconstruction of the corresponding output states~\cite{d2003imprinting}. Indeed, SPT can be seen as an extreme case of such a situation, where the presence of an ancilla is actually irrelevant, and one just uses a faithful set of probe states. We remark that the correlations present in one or more of the bipartite states of a faithful set can be deemed as effectively assisting process tomography as long as the faithful set comprises less than $d^2$ states. The results in this paper lie between the two archetypical techniques sketched above, and focus on the exploitation of correlations to reduce the number of distinct inputs needed for what we could call in general \emph{correlation-assisted process tomography} (CAPT).

More specifically, we focus on the question of whether and how a faithful set can be generated by means of local actions $\{\Gamma_i\}$ on a fixed input state (see Figure~\ref{fig:capt}).
\begin{figure}[b]
\[
\Qcircuit @C=1em @R=.7em @! R {    
 & \multiprepareC{1}{\rho_{AB}}	& \ustick{A}\qw			& \gate{\{\Gamma_i\}}	& \gate{\Lambda}	& \multimeasureD{1}{M}\\
 & \pureghost{\rho_{AB}}  		& \ustick{B}\qw			& \qw			& \qw			& \ghost{M}		}	
\]
\caption{Correlation-assisted quantum process tomography. Any bipartite state $\rho_{AB}$ can be used in this scheme. The presence of correlations in the state may substantially reduce the number of known channels $\{\Gamma_i\}$ that need to be applied so that $\{\Gamma_{i,A}[\rho_{AB}]\}$ is a faithful set. Standard process tomography and ancilla-assisted process tomography are extreme cases of this more general scenario.}
\label{fig:capt}
\end{figure}
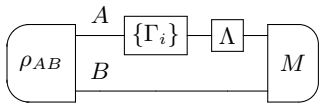
In the case where there is no ancilla (or, if there is an ancilla, where there are no probe-ancilla correlations), a local action is not very different from simply considering $d^2$ inputs, but one may need strictly less than $d^2$ local operations on the input if correlations are present between probe and ancilla. That is, our results may be interpreted as an interpolation between the use of fully uncorrelated or fully correlated (that is, having maximum OSR) input states.

We show in general that a faithful set can always be generated via $\left\lceil\frac{d^2}{\OSR(\rho_{AB})}\right\rceil$ local transformations on a \emph{fixed bipartite state}. Notice that this is optimal, as it is clearly impossible to generate a faithful set with less local operations. We also consider the case where such local transformations are constrained to be unitary. 
For pure fixed states, we find that such a constraint does not change the result: any pure bipartite state of Schmidt rank $k$ (hence with OSR equal to $k^2$) can be used to generate a faithful set with 
$\left\lceil\frac{d^2}{k^2}\right\rceil$ unitaries. 
For mixed states, the constraint can actually be limiting: we exhibit a class of qudit-qudit states with OSR equal to two but such that one still needs $d^2$ local unitaries to generate a faithful set. We conjecture that in general one may need $\left\lceil\frac{d^2}{\OSR(\rho_{AB})-1}\right\rceil$ local unitaries to generate a faithful set. On the other hand, the mixed-state case can display a highly ``efficient'' (in terms of local unitaries employed) generation of a faithful set; specifically, we exhibit a family of qudit-qudit states with OSR $\approx d^2/2$ where only two unitaries are needed to achieve faithfulness; notice that such a case is impossible in the pure-state case. Finally, by exploiting the relation between $SO(3)$ and $SU(2)$ (that is, in a sense, the Bloch ball qubit representation), we fully characterize the qubit-qudit case for qubit channels, once more highlighting the importance of discord in the issue of correlation-assisted channel tomography/discrimination: a two-qubit state gives rise to a faithful set with at most two local unitaries if an only if it exhibits discord on the probe side.

\noindent
\section{Notation and preliminaries}%
We will focus on finite-dimensional systems. Any Hilbert space $\cH$ will be equivalent to $\mathbb{C}^d$, for some dimension $d$. We will indicate by $\cL(\cH)$, the set of linear operators $L$ (equivalent to matrices) on $\cH$. The dimension of a system $X$ with Hilbert space $\cH_X$ will be indicated by $d_X$. A quantum state on $\cH$ is a density operator $\rho$ belonging to the convex subset $\cD(\cH)\subset \cL(\cH)$ of positive semidefinite operators with unit trace. The trace operation is indicated by $\Tr$, while we denote by $\Tr_X$ the partial trace on system $X$.
We write $\rho_X$ for the (reduced) state of system $X$.

A bipartite state $\rho_{AB}$ is unentangled (or separable)
if it is the convex combination of product (or uncorrelated) states~\cite{revent}, that is, if
$\rho_{AB} = \sum_i p_i \rho^A_i \otimes \rho^B_i$, 
with $\{p_i\}$ a probability distribution. A state is considered to be entangled if and only if it is not separable. 

Any physical evolution (or quantum process) affecting a quantum system is described in terms of quantum channels~\cite{nielsen2010quantum}. 
A quantum channel from $X$ to $Y$ is a completely-positive trace-preserving linear map $\Lambda$ from $\cL(\cH_X)$ to $\cL(\cH_Y)$.

We will also make use of some specific norms on operators, in particular we will use (some of) the $p$-norms
\[
\|L\|_p = (\sum_i \sigma_i(L)^p)^{1/p},
\]
where the sum is over the singular values $\sigma_i(L)$ of $L$. The trace norm $\|L\|_1=\sum_i \sigma_i(L)$ corresponds to $p=1$; the Frobenius norm $\|L\|_p = \sqrt{\sum_i \sigma_i(L)^2}$ corresponds to $p=2$; finally, for $p\rightarrow \infty$, we have the operator norm $\|L\|_\infty=\max_i \sigma_i(L)$.

For simplicity, we will typically consider bipartite states $\rab$ with the two subsystems of the same dimension, and we define $d\deff d_A=d_B$, keeping in mind that the results and terminology (e.g., ``maximal OSR'') presented in this paper are generally valid as long as $d_A\leq d_B$.

\section{Operator Schmidt decomposition}

The state vector $\ket{\psi}_{AB}\in \cH_A\otimes \cH_B$ of a bipartite system admits a Schmidt decomposition~\cite{nielsen2010quantum} $\ket{\psi}_{AB} = \sum_{i=1}^{d_\textup{min}} \sqrt{p_i}\ket{a_i}_A\otimes\ket{b_i}_B$,
with $d_\textup{min}=\min\{d_A,d_B\}$ and $\{p_i\}$ positive numbers that sum up to 1, thus constituting a probability distribution.
We will take $\{p_i\}$ to be ordered without loss of generality, that is, we will assume $p_1\geq p_2 \geq\ldots$. The sets of orthonormal vectors $\{\ket{a_i}\}$ and $\{\ket{b_i}\}$ are some special and $\ket{\psi}$-dependent orthonormal bases for $\cH_A$ and $\cH_B$, respectively. We denote by $\SR(\psi)$  the Schmidt rank of $\ket{\psi}_{AB}$; it is the number of non-zero $p_i$'s, and satisfies $\SR(\psi)\leq d_{\min}$.

By considering a bipartite state $\rho_{AB}$ as a vector in $\cL(\cH_A\otimes\cH_B)$, we arrive at the Operator Schmidt Decomposition (see~\cite{aniello2009relation,lupo2008bipartite} and references therein)
\begin{equation}
\label{eq:OSD}
\rho_{AB} = \sum_{i=1}^{\OSR(\rho)}r_i A_i \otimes B_i.
\end{equation}
We denote by $\OSR(\rho)$ the Operator Schmidt Rank (OSR) of $\rho$; this is the number of non-zero Operator Schmidt Coefficients (OSCs) $r_i$, which are taken to be positive. The sets of operators  $\{A_i\}_{i=1}^{d_A^2}$ and $\{B_i\}_{i=1}^{d_B^2}$ form ($\rho$-dependent) orthonormal bases for the spaces $\cL(\cH_A)$ and $\cL(\cH_B)$, respectively. We remark that, since $\rho_{AB}$ is Hermitian, the two orthonormal operator bases in~\eqref{eq:OSD} can be (but need not be) taken to be composed of Hermitian operators. The OSCs are the singular values of the matrix $[C_{ij}(\rho_{AB})]_{ij}$, with $C_{ij}(\rho_{AB}):=\Tr( F_i^\dagger\otimes G_j^\dagger \rho_{AB})$. Such a matrix is sometimes referred to as \emph{correlation matrix}. The operators $\{F_i\}$ and $\{G_j\}$ are arbitrary local orthonormal bases for the respective operator spaces. As in the case of the Schmidt decomposition for vector states, without loss of generality we can take the OSCs to be ordered as $r_1 \geq r_2 \geq \ldots $;
Such coefficients satisfy $\sum_ir_i^2 = \Tr(\rho^2)$. One immediately convinces oneself that $\OSR(\rho_{AB})\leq d_{\min}^2$, as the vector space $\cL(\cH_A)$ has dimension $d_A^2$ (similarly for $\cL(\cH_B)$). The SD of a vector state $\ket{\psi}_{AB}$ and the OSD of the corresponding density matrix $\proj{\psi}_{AB}$ are very much related:  one has $r_i=\sqrt{p_k}\sqrt{p_l}$, $A_i = \projj{a_k}{a_l}$, and $B_i = \projj{b_k}{b_l}$, for $i=(k,l)$ a multi-index, so that $\OSR(\proj{\psi}_{AB})=\SR(\psi_{AB})^2$.

\section{Operator Schmidt Decomposition and ancilla-assisted process tomography}

As anticipated, AAPT requires the preparation of a bipartite system in a single bipartite state $\rab$. One subsystem (the probe) is sent through the channel $\Lambda$ to be characterized. Using \reff{eq:OSD}, the output $\rho_{\Lambda}\deff(\Lambda\otimes\id)[\rab]$ reads
\begin{equation}\label{output}
\rho_{\Lambda}=\sum_{l=1}^{OSR(\rab)}r_l\Lambda[A_l]\otimes B_l.
\end{equation} 
Then, by reconstructing $\rho_\Lambda$, one recovers the action of the channel on the basis element $A_l$ via $\Lambda[A_l]=\Tr_B((\mathbb{1}\otimes B_l^\dag)\rho_{\Lambda})/r_l$ (for $r_l>0$). It follows that inputs with maximal OSR enable complete characterization of the channel, since its action on a complete operator basis of $\cL(\cH_A)$ can be reconstructed \cite{d2003imprinting,altepeter2003ancilla}. It is clear that input states defined as faithful are states with maximal OSR, more precisely with OSR equal to $d_A^2$. Correspondingly, a set of (potentially unfaithful, when considered individually) bipartite states $\{\rho_{AB,i}\}$ is called faithful if the local operators $\{A_{l,i}\}_{l=1}^{\OSR(\rho_{AB,i})}$ of $\OSD(\rho_{AB,i})$, when considered together, generate the whole $\cL(\cH_A)$, i.e., if $\text{span}(\{A_{l,i}\}_{l,i})=\cL(\cH_A)$.

\section{Generating a faithful set of inputs with general local channels}

The core idea of this work is to show that the correlations of a fixed bipartite state, of whatever degree,  can in principle be exploited to allow ``more efficient'' process tomography. Such correlations can be measured through the OSR, which obeys very general requirements that any meaningful measure of total correlations must satisfy. Indeed, the OSR is minimal for and only for non correlated (product) states, and it is monotone under local channels. 
Not relying on correlations, like in SPT, is the same as considering minimal OSR---that is, OSR equal to 1---for the inputs. On the other hand, fully relying on correlations, like in AAPT, means requiring maximal OSR for a single bipartite input. These considerations legitimate the intuition that intermediate values for the OSR should be consistent with the use of an intermediate number of inputs.

In this section we analyze how we can achieve the condition $\text{span}(\{A_{l,i}\})=\cL(\cH_A)$ indicated at the end of the last section, where each set $\{A_{l,i}\}_l$ is a local orthonormal basis for $(\Gamma_i\otimes\id)[\rho_{AB}]$, for $\Gamma_i$ the local operations applied on the probe before it is subject to the channel $\Lambda$ (see Figure \ref{fig:capt}).

Let us first remark why channel tomography is certainly possible in this setup. The reason is simple: each channel $\Gamma_i$ may simply be taken to have constant output corresponding to one of the input states $\rho_A^i$ used in standard channel tomography (Figure \ref{fig:SPT}). With this ``trivial'' strategy, we do not make use of correlations at all, but we certainly achieve the task at hand. Having estalished this, let us move to the issue of ``optimizing'' the $\Gamma_i$'s, at least with respect to their number.

Let $\{A_l\}_{l=1}^{d^2}$ be a Hermitian local orthonormal basis for the operator Schmidt decomposition for $\rho_{AB}$, comprising the $\OSR(\rho_{AB})$ elements corresponding to non-zero OSCs. It is clear that $\vspan(\{A_{l,i}\}_l) = \vspan(\{\Gamma_i[A_{l}]\}_l)$, so that $\text{span}(\{A_{l,i}\}_{l,i}) = \vspan(\{\Gamma_i[A_{l}]\}_{l,i})$. Thus, our goal is the following: given $\{A_l\}_{l=1}^{\OSR(\rho_{AB})}$, find a (minimal) way of choosing the local maps $\Gamma_i$ so that  $\vspan(\{\Gamma_i[A_{l}]\}_{l,i}) = \cL(\cH_A)$. By minimal, we mean that we want to identify the smallest possible number of local channels $\Gamma_i$ that are needed to achieve such a condition. In the following we provide a construction to achieve this.

Let us consider the following family of maps,
\begin{equation}\label{eq:tp_cp}
\Gamma_i[X]\deff (1-\epsilon)\Tr(X)\frac{\mathbb{1}}{d}+\epsilon\, \tilde\Gamma_i[X],
\end{equation}
which are each a convex combination (for $0\leq \epsilon \leq1$) of the totally depolarizing channel $X\mapsto \Tr(X)\frac{\mathbb{1}}{d}$ and of
\begin{equation}
\begin{split}
\tilde\Gamma_i[X]&\deff\sum_j\Tr(A_jX)A_{\gamma_{j,i}} \\
&\quad+[\Tr(X) -\sum_{j}\Tr(A_jX)\Tr(A_{\gamma_{j,i}})]\frac{\mathbb{1}}{d}.
\end{split}
\label{eq:tp_ncp}
\end{equation}
Notice that $\tilde\Gamma_i$ is not necessarily a channel, but it is a linear map that is trace preserving by construction. 
Here we denote $\gamma_{j,i}\deff j\oplus (i\cdot \OSR(\rab))$, and the $A_i$'s form a local orthonormal basis that is a superset of the local operators of the OSD of $\rho_{AB}$, as in \reff{eq:OSD}. The symbol $\oplus$ indicates addition modulo $d^2$, and we let $i=1,2,3\ldots$. By construction, the maps $\Gamma_i[X]$ are trace preserving, and, for $0<\epsilon<1$ small enough, completely positive. This is because, within the set of linear trace-preserving maps, there is a ball of completely positive maps around the totally depolarizing channel. Notice that, in principle, we could consider any other channel with full-rank fixed output, at the ``cost'' of considering some other $\epsilon$. Such full-rank fixed output (as well as $\epsilon$) could even be made to depend on $i$.

It is easy to recognize the action of $\Gamma_i$ on the generic basis element $A_n$:
\begin{equation}
\Gamma_i[A_n]= \big(\Tr(A_n-\epsilon A_{\gamma_{n,i}})\big)\frac{\mathbb{1}}{d}+\epsilon\, A_{\gamma_{n,i}},
\end{equation}
i.e. the $i-$th map, acting upon the $n$-th element of the local basis, returns a liner combination of the basis element indexed as $n\oplus (i\cdot \OSR(\rab))$ and of the identity. To make the action of the channels clearer, let us consider the action of, e.g.,  $\Gamma_1$. The latter would map the set $\{A_n\}_{n=1}^{\OSR(\rho_{AB})}$ to 
$\{
	\big(
		\Tr(A_n-\epsilon A_{n\oplus \OSR(\rho_{AB})})
		\big)
		\frac{\mathbb{1}}{d}
		+
		\epsilon\,
		A_{n\oplus \OSR(\rho_{AB})}
\}_{n=1}^{\OSR(\rho_{AB})}.
$
Thus, it should be clear that---up to the detail of whether we need to choose a fixed state different from the maximally mixed one in \eqref{eq:tp_cp} and \eqref{eq:tp_ncp} in order to certainly obtain a set whose elements are all linearly independent from the other generated sets---given an incomplete set of basis elements $\{A_n\}_{n=1}^{\OSR(\rho_{AB})}$, we are able to obtain operators spanning the same space as the remaining $d_A^2-\OSR(\rho_{AB})$ ones through the application of $\lceil\frac{d_A^2}{\OSR(\rho_{AB})}\rceil-1$ channels. This leads directly to the following theorem.

\begin{Theorem}
	Let $\rab$ be a bipartite state with $OSR(\rab)=k$. Then there is a set of quantum channels $\Gamma_i$, with $i=1,\ldots,\lceil\frac{d_A^2}{k}\rceil$, such that $\{(\Gamma_i\otimes\id)[\rho_{AB}]\}_i$ is faithful. Without loss of generality, one of the channels can be taken to be the identity channel.
\end{Theorem}

\section{Generating a faithful set of inputs with unitary local channels}

In this section we consider constraining the local channels $\Gamma_i$ that act on the probe in Fig.~\ref{fig:capt} to be unitary. The question we address is that of determining how many unitary rotations $U_i$ are needed in order to obtain a faithful set of input states $\{U_i\otimes\I\rho_{AB}U_i^\dagger\otimes\I\}_i$. As discussed in the case of general local operations, this corresponds to finding out how many unitary rotations are needed so that  $\{U_iA_lU_i^\dagger\}_{i,l}$ spans the entire space $\cL(\cH_A)$, where $\{A_l\}_l$ is a set of orthonormal local OSD operators corresponding to non-zero operator Schmidt coefficients.

We remark that in the case where we impose the constraint that the channels be unitary, the fact that process tomography is possible at all is not immediate. Indeed, it is not anymore the case that this is possible for \emph{all} input states $\rho_{AB}$. Nonetheless, we prove that it is possible for all states that are not of the form $\frac{{\I}_A}{d}\otimes\rho_B$: notice that the latter states are not only uncorrelated, but such that the state of the probe $A$ is maximally mixed. Notice also that our result means that any form of correlations is enough to make process tomography by local unitaries possible.

Let us first establish this result.

\subsection{Process tomography via local unitary rotation of almost any input}

It is convenient to recall the definition of frame \cite{duffin1952class}. Such a concept is generally defined for families of vectors in inner product spaces. In our framework we exploit the inner product structure of $\cL(\cH)$ and define a frame as a collection  $\{P_k\}$ of operators such that there are real numbers $0<a\leq b<\infty$ satisfying
\begin{equation}\label{fc}
	a\|X\|_2^2\leq\sum_k|\Tr(P_k X)|^2\leq b\|X\|_2^2
\end{equation}
for any $X\in \cL(\cH)$. A frame generalizes the notion of basis. Notice in particular that, if the frame is actually an orthonormal basis, that is, if $\Tr(P_k^\dagger P_l)=\delta_{kl}$, then the frame condition \eqref{fc} is satisfied with $a=b=1$. In finite dimensions, a finite collection $\{P_k\}$ is a frame for $\cL(\cH)$ if and only if it is a spanning set for $\cL(\cH)$, while an infinite collection $\{P_k\}$, even when a spanning set, may not constitute a frame, as there might not be a finite $b$ that satisfies \eqref{fc}. The lower bound in \eqref{fc}, for $a>0$, ensures that $X$ can be reconstructed from the values $\Tr(P_k X)$. It should be clear that, given a frame, one can always consider a subset of the elements of the frame, so that such subset forms a basis, that is, a spanning set of linearly independent operators.

What we will prove is that 
it is possible to choose $d^2$ unitaries $\{U_i\}_{i=1}^{d^2}$ (with $U_1=\I$ without loss of generality) so that $\{U_i\rho U_i^\dagger\}_i$ is a frame and a basis for the space of operators of the input ancilla, initially prepared in the state $\rho$, as long as $\rho\not\propto \I$. To prove this, we will need the notion of twirling, or 
twirl operation \cite{werner1989quantum}.
The latter is the linear projection $\mathcal{T}$ on bipartite operators $Y\in\cL(\mathbb{C}^d\otimes\mathbb{C}^d)$ defined as 
\begin{equation}\label{twirl}
	\mathcal{T}(Y)=\int_U  (U\otimes U)Y(U^{\dagger}\otimes U^{\dagger})dU,
\end{equation}
where the integral is taken with respect to the Haar measure of the unitary group on $\mathbb{C}^d$. Since any operator commuting with all unitaries of the form $U\otimes U$ can be written as a linear combination of $\mathbb{1}$ and $V$ (where $V$ is the flip operator, defined implicitly by the its action $V\ket{\psi}\otimes\ket{\varphi}=\ket{\varphi}\otimes\ket{\psi}$, for all $\ket{\psi},\ket{\varphi}\in \mathbb{C}^d$), it follows that \cite{werner1989quantum}
\[
\mathcal{T}(Y)=\alpha(Y)\mathbb{1}+\beta(Y) V.
\]
The coefficients $\alpha(X)$ and $\beta(X)$ are fixed by the conditions
\begin{align}
\Tr(\mathcal{T}(Y)) 	&=	\Tr(Y)	\\
\Tr(\mathcal{T}(Y)V)	&=	\Tr(YV),
\end{align}
solved by
\begin{align}
\label{eq:alphabeta}
\alpha(Y)&=\frac{d\Tr(Y)-\Tr(YV)}{d^3-d}, \\ \beta(Y)&=\frac{d\Tr(YV)-\Tr(Y)}{d^3-d}.
\end{align}
We will use the fact that the twirling can be approximated by a unitary $2$-design, that is by a finite set of $n$ unitaries $\{U_i\}_{i=1}^n$ such that
 \begin{equation}\label{twirldesign}
 \mathcal{T}(Y)= \frac{1}{n}\sum_{i=1}^n U_i\otimes U_i Y U_i^{\dagger}\otimes U_i^{\dagger}.
 \end{equation}
We will obtain our frame by taking $d^2$ of such unitaries.
 
Let $\{U_i\}_{i=1}^n$ be a unitary $2$-design (without loss of generality, one of the unitaries can be taken to be the identity). Let us check the frame conditions \eqref{fc} of $\{U_iA U_i^\dagger\}_{i=1}^n$, for an arbitrary $A\in\cL(\mathbb{C}^d)$. One has
\begin{align}
&\quad\sum_{i=1}^{n}|\Tr(U_iAU_i^\dag X)|^2 \\
&=\sum_{i=1}^{n}\Tr(U_iAU_i^\dag X)\Tr(U_iA^\dagger U_i^\dag X^\dag)\nonumber\\
&=\sum_{i=1}^{n}\Tr\Big((U_i\otimes U_i)(A\otimes A^\dag)(U_i^\dag\otimes 
U_i^\dag) (X\otimes X^\dag)\Big)\nonumber\\
&\propto \Tr(\mathcal{T}(A\otimes A^\dag) X\otimes X^\dag)\nonumber
\\
&=\Tr((\alpha(A\otimes A^\dagger)\mathbb{1}+\beta(A\otimes A^\dagger) V) X\otimes X^\dag)\nonumber\\
&=\alpha(A\otimes A^\dagger)\Tr( X\otimes X^\dag)+\beta(A\otimes A^\dagger)\Tr(VX\otimes X^\dag)\nonumber\\
&=\alpha(A\otimes A^\dagger)|\Tr(X)|^2+\beta(A\otimes A^\dagger)\|X\|_2^2\\
\end{align}
where 
$\alpha(A\otimes A^\dagger)$ and $\beta(A\otimes A^\dagger)$ are given by Eqs. \eqref{eq:alphabeta} applied to the case $Y=A\otimes A^\dagger$, so that 
\begin{align}
\alpha(A\otimes A^\dagger)&=\frac{d|\Tr(A)|^2-\|A\|_2^2}{d^3-d}, \\ \beta(A\otimes A^\dagger)&=\frac{d\|A\|_2^2-|\Tr(A)|^2}{d^3-d}.
\end{align}
Working in finite dimensions, we see that the frame condition \eqref{fc} is achieved as long as $\alpha(A\otimes A^\dagger)>0$, $\beta(A\otimes A^\dagger)>0$, which means as long as
\begin{equation}
d|\Tr(A)|^2\geq\|A\|_2^2
\label{eq:cond1}
\end{equation}
and
\begin{equation}
|\Tr(A)|^2<d\|A\|_2^2.
\label{eq:cond2}
\end{equation}
Let us  assume that $A$ is a state, specifically the reduced state $\rho_A$ of the probe. Then, the first inequality is automatically satisfied. Moreover, the Cauchy-Schwarz inequality implies that $|\Tr(A)|^2\leq |\Tr(\I)|\|A\|_2^2=d|\|A\|_2^2$, with equality if and only if $A\propto \I$. 
Having assumed that that $A$ is the state $\rho_A$, this is the condition that $\rho_A$ is not maximally mixed.

Thus, we have found that, independently of the presence of an ancilla, as long as the reduced state $\rho_A$ of the probe  is not maximally mixed, we can find $d^2$ unitaries, one of which is the identity, such that  $\{U_i\rho_A U_i^\dagger\}_{i=1}^{d^2}$ is a tomographically faithful set.

We can extend this result to the case where there are non-vanishing correlations. The operator $A$ for which we want that $\{U_iA U_i^\dagger\}_{i=1}^{d^2}$ be a tomographically complete set can be taken to be any linear combination of the hermitian operator Schmidt operators $\{A_l\}_l$ corresponding to non-zero operator Schmidt coefficients (one example being the reduced state $\rho_A$). Suppose $\rab$ is not product. Then there are at least two terms in its OSD, and at least one between one of the $A_i's$ and the reduced state $\rho_A$ is not proportional to the identity; we can then consider as $A$ in the above construction of the frame some linear combination of the latter two operators that respects conditions \eqref{eq:cond1}-\eqref{eq:cond2}.
On the other hand, since a state of the form $\frac{{\I}_A}{d}\otimes\rho_B$  is invariant under local unitary transformation on $A$, we have proven our statement:

\begin{Theorem}\label{th:OSR1} For all non product bipartite states and for all product states $\rab$ such that $\ra\neq\mathbb{1}_A/d$, there always exist $d^2$ unitary operators $U_i\in SU(d)$ such that the set $\{(U_i\otimes\id)\rab(U_i\otimes\id)^\dag\}_{i=1}^{d^2}$ is faithful.
\end{Theorem}

\noindent 
\subsection{Pure probe-ancilla state}
 
For pure states we are able to find the optimal number of local unitaries needed to construct a faithful set starting from a fixed pure state of Schmidt rank $k$:
\begin{equation}\label{eq:PS}
\proj{\psi}_{AB}=\sum_{i,j=1}^k\sqrt{p_ip_j}\projj{i}{j}_A\otimes\projj{i}{j}_B.
\end{equation} 
\begin{Theorem}\label{th:PS}
Let $\proj{\psi}_{AB}$ be as in \reff{eq:PS}. Then, there are $n\deff\left\lceil\frac{d}{k}\right\rceil^2$ local unitaries $U_i$ such that the set given by $\{(U_i\otimes\I)\proj{\psi}_{AB}(U_i\otimes\I)^\dag\}_{i=0}^{n-1}$, with $U_0=\I$, is faithful.
\end{Theorem}
\begin{proof}
Let $\ket{\psi_0}=\ket{\psi}$. State tomography of the output $(\Lambda\otimes\id)\proj{\psi_{0}}$  determines the channel $\Lambda$ partially, i.e., its action on $\{\projj{i}{j}\}$ only for $i,j=1,\ldots,k$. To obtain the image under $\Lambda$ of the remaining $\projj{i}{j}$ elements, it is convenient to consider the case when $k$ divides $d$.

We will start by analyzing how it is possible to reconstruct the action of $\Lambda$ on all of $\{\projj{i}{j}\}$, for $i,j=1,\ldots,2k$. 
Let us define the set of operators
\begin{align}
A_{ij}&=\projj{i}{j},\\
B_{ij}&=\projj{i+k}{j+k},\\
C_{ij}&=\projj{i+k}{j},\\
D_{ij}&=\projj{i}{j+k},
\end{align}
for $i,j=1,\ldots,k$, and where sums within kets should be in general understood modulus $d$.
Also, let us introduce unitary operators whose action restricted to the vectors $\ket{n}$, for $n=1,\ldots,k$, is given by
\begin{align*}
&X\ket{n}=\ket{n+k},\\
&U\ket{n}=2^{-1/2}(\ket{n}+\ket{n+k}),\\
&V\ket{n}=2^{-1/2}(\ket{n}+i\ket{n+k}).\\
\end{align*}
Acting locally on $\proj{\psi}$, such operators produce the following states
\begin{align*}
\proj{\psi_1}\vcentcolon=&(X\otimes\mathbb{1}) \proj{\psi} (X\otimes\mathbb{1})^\dag\\
\proj{\psi_2}\vcentcolon=&(U\otimes\mathbb{1}) \proj{\psi} (U\otimes\mathbb{1})^\dag\\
\proj{\psi_3}\vcentcolon=&(V\otimes\mathbb{1}) \proj{\psi} (V\otimes\mathbb{1})^\dag.
\end{align*}
Define $\Lambda[Y]=[\Lambda[Y_{ij}]]_{i,j=1}^{k}$, for $Y=A,B,C,D$. Then $\Lambda[A]$ is reconstructed  through tomography of $(\Lambda\otimes\id )\proj{\psi_{0}}$ (as already noticed), while $\Lambda[B]$ is obtained from $(\Lambda\otimes\id )\proj{\psi_1}$. On the other hand, $\Lambda[C]$ and $\Lambda[D]$ can be reconstructed by measuring the four outputs (i.e. $(\Lambda\otimes\id )\proj{\psi_{l}}$ for $l=0,\ldots,3$) and then combining the results. To be more precise, since 
\begin{align*}
C_{ij}	&=UA_{ij}U^\dag+iVA_{ij}V^\dag\\
		&\quad-\frac{1+i}{2}(A_{ij}+XA_{ij}X^\dag)\\
D_{ij}&=iVA_{ij}V^\dag-UA_{ij}U^\dag\\
		&\quad-\frac{i-1}{2}(A_{ij}+XA_{ij}X^\dag),
\end{align*}
\noindent
linearity implies 
\begin{align*}
\Lambda[C_{ij}]&=\Lambda[UA_{ij}U^\dag]+i\Lambda[VA_{ij}V^\dag]\\
&\quad-\frac{1+i}{2}(\Lambda[A_{ij}]+\Lambda[XA_{ij}X^\dag])\\
\Lambda[D_{ij}]&=i\Lambda[VA_{ij}V^\dag]-\Lambda[UA_{ij}U^\dag]\\
&\quad-\frac{i-1}{2}(\Lambda[A_{ij}]+\Lambda[XA_{ij}X^\dag]).
\end{align*}
Thus, we see that we have reconstructed $\Lambda[\projj{i}{j}]$ for $i,j=1,\ldots,2k$ with four local unitaries.

Information on the remaining $\Lambda[\projj{i}{j}]$ can be reconstructed similarly.
The theorem follows by reiterating this procedure, until recovering the action of $\Lambda$ on all the blocks.
More explicitly, it is possible to reconstruct $\Lambda[\projj{i}{j}]$ for $i,j\in \{1+p\cdot k,\dots, k+p\cdot k\}$ and $p=0,\ldots,d/k -1$
by considering the action of $p$-labelled $d/k$ unitaries (one being the identity) each performing one of the transformations
\begin{align*}
&\ket{n}\mapsto\ket{n+p\cdot k}.\\
\end{align*}
Once these `on-diagonal blocks' have been reconstructed, it is then possible to further reconstruct the `off-diagonal blocks' $\Lambda[\projj{i}{j}]$, for $i\in \{1+p\cdot k,\dots, k+p\cdot k\}$ and $j\in \{1+q\cdot k,\dots, k+q\cdot k\}$, $p\neq q$ by the use of $(d/k(d/k-1))/2$ pairs of unitaries that perform the transformations
\begin{align*}
&\ket{n}=2^{-1/2}(\ket{n+p\cdot k}+\ket{n+q\cdot k}),\\
&\ket{n}=2^{-1/2}(\ket{n+p\cdot k}+i\ket{n+q\cdot k}).\\
\end{align*}
This gives a total of $d/k + 2\cdot(d/k(d/k-1))/2 = (d/k)^2$ unitaries.

If $k$ does not divide exactly $d$, then one needs to consider an additional set of unitaries, but obviously the cost (in terms of unitaries) cannot be larger than in the case where we imagine the $A$ system embedded in a $d'$-dimensional system, with $d'=\lceil \frac{d}{k} \rceil\cdot k$.
\end{proof}
\noindent
In the light of the last theorem we see that the higher the correlations (in terms of OSR) of the fixed pure state, the less $U_i$ are required. As expected, when the fixed pure state has maximal OSR, one recovers completely the AAPT scenario. For pure state with $\OSR=1$, the number of experimental settings to perform channel tomography is again the one of SPT. As a final remark we observe, by looking at the proof of Theorem \ref{th:PS}, that one can derive the specific form of a particular set of $U_i$, besides establishing their existence. 

Contrary to the pure state case, for the case where the fixed state is mixed we have not derived a formula which directly links the OSR of the input to the number of unitaries needed to reach faithfulness. However, in the following we give specific examples that show that also when the fixed state is mixed, the presence of correlations dramatically reduces the number of local unitaries required to perform channel tomography. 

\subsection{Mixed probe-ancilla state: qubit-qudit inputs}

The first example involves a qubit-qudit system, for qubit channel tomography. We show that reducing the cardinality of the faithful set created by local unitaries on the qubit depends strongly on the quantumness of correlations on the qubit side. Before going into the details it is convenient to recall that a bipartite state is called classical on $A$ if it can be expressed as $\rab = \sum_i p_i \proj{a_i}_A\otimes \rho_i^B$, for some orthonormal basis $\{\ket{a_i}\}$. States that are not classical on $A$ are said to have non-zero quantum discord~\cite{henderson2001classical,ollivier2001quantum,RevModPhys.84.1655}. Also, we will make use of the following Lemma, in which we use the notion of Bloch vector for a generic Hermitian operator $L=L^\dagger$, given by $\vec{l}=(l_1,l_2,l_3)$, with $l_i=\Tr(\sigma_i L)$ and $\sigma_i$, $i=1,2,3$, the Pauli operators.
\begin{Lemma}
	\label{lem:commuting}
Consider Hermitian operators $A,B\in\cL(\mathbb{C}^2)$. Then, $A$ and $B$ commute if and only if their Bloch vectors are proportional.
\end{Lemma}
\begin{proof}
Let $\sigma_0=\mathbb{1}$ and denote $a_0=\Tr(A)$, $b_0=\Tr(B)$. Let also $\vec{a}=(a_1,a_2,a_3)$ and $\vec{b}=(b_1,b_2,b_3)$ be the Bloch vectors of $A$ and $B$, respectively, so that $A=\frac{1}{2}\sum_{i=0}^3a_i\sigma_i$ and $B=\frac{1}{2}\sum_{i=0}^3b_i\sigma_i$. 
Observe that
\begin{align*}
[A,B]
&=\left[\frac{1}{2}\sum_{i=0}^3a_i\sigma_i,\frac{1}{2}\sum_{i=0}^3b_i\sigma_i\right]\\
&=\frac{1}{4}\sum_{i,j=0}^3 a_ib_j  \left[\sigma_i,\sigma_j\right]\\
&=\frac{1}{4}\sum_{i,j,k=1}^3a_ib_j\,  2 i\epsilon_{ijk}\sigma_k\\
&=\frac{i}{2}\sum_{k=1}^3\left(\sum_{i,j=1}^3a_ib_j\epsilon_{ijk}\right)\sigma_k\\
&=\frac{i}{2}(\vec{a}\times\vec{b})\cdot\vec{\sigma}
\end{align*}
where we used the Levi-Civita symbol $\epsilon_{ijk}$, and $\times$ indicates the standard cross product between three-dimensional vectors.
\noindent
Since $\sigma_1,\sigma_2,\sigma_3$ are linearly independent, the expression in the last line above is zero if and only if the cross product $\vec{a}\times\vec{b}$ vanishes, which happens if and only if $\vec{a}=\lambda\vec{b}$, with $\lambda\in\mathbb{R}$.
\end{proof}

We are now in the position to state the following.
\begin{Theorem}
Let $A$ be a qubit.
Then, $\rab$ has quantum discord on $A$ if and only if $\rab$ allows correlation-assisted process tomography on $A$ with at most two unitary rotations.
\end{Theorem}
\begin{proof}
We recall that a qubit-qudit state has zero discord on the qubit side $A$ if and only if $\rab=p\proj{a_1}_A\otimes\rho_1^B+(1-p)\proj{a_2}_A\otimes\rho_2^B$,
with $\{\ket{a_1},\ket{a_2}\}$ some orthonormal basis for $A$. While this is not necessarily the operator Schmidt decomposition, it is clear that the state only allows to reconstruct the action of a channel $\Lambda$ on
$\vspan(\{\proj{a_1},\proj{a_2}\})=\vspan(\{\I,\proj{a_1}\})$. Overall, a single additional unitary rotation $U$ allows us to reconstruct only the action of the same map on
$\vspan(\{\proj{a_1},\proj{a_2},U\proj{a_1}U^\dag,U\proj{a_2}U^\dag\})=\vspan(\{\I,\proj{a_1},U\proj{a_1}U^\dag\})$, which is not enough to tomographically reconstruct the channel. A geometric way of thinking about this is that the resulting four Bloch vectors are necessarily coplanar, and do not span affinely $\mathbb{R}^3$ (see Fig. \ref{fig:nodiscord}). 

\begin{figure}
	\begin{subfigure}[b]{0.5\textwidth}
		\includegraphics[width=0.4\textwidth]{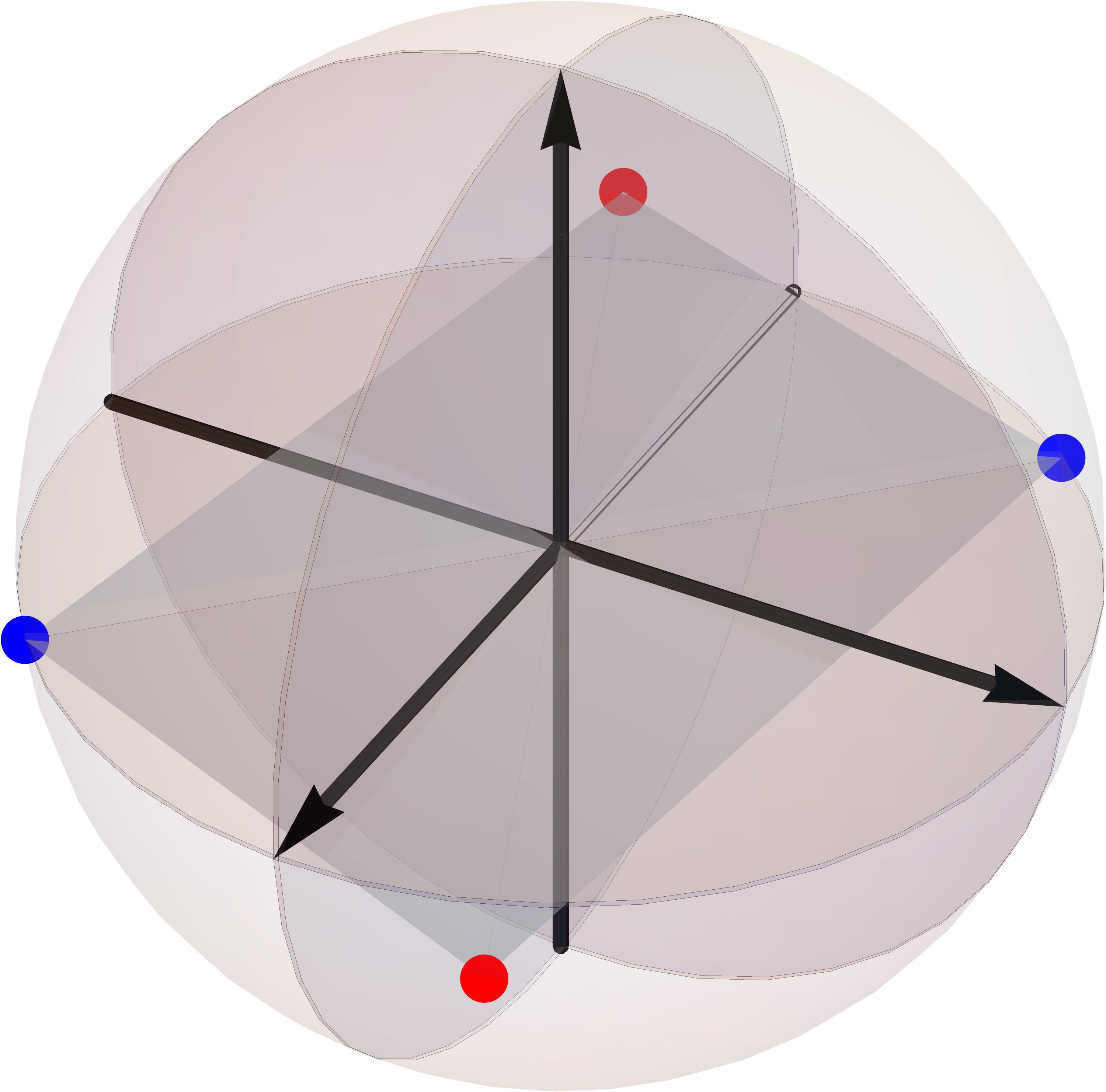}
		\caption{With no discord.}
		\label{fig:nodiscord}
	\end{subfigure}
	\vspace{10pt}
	
	\begin{subfigure}[b]{0.5\textwidth}
		\includegraphics[width=0.4\textwidth]{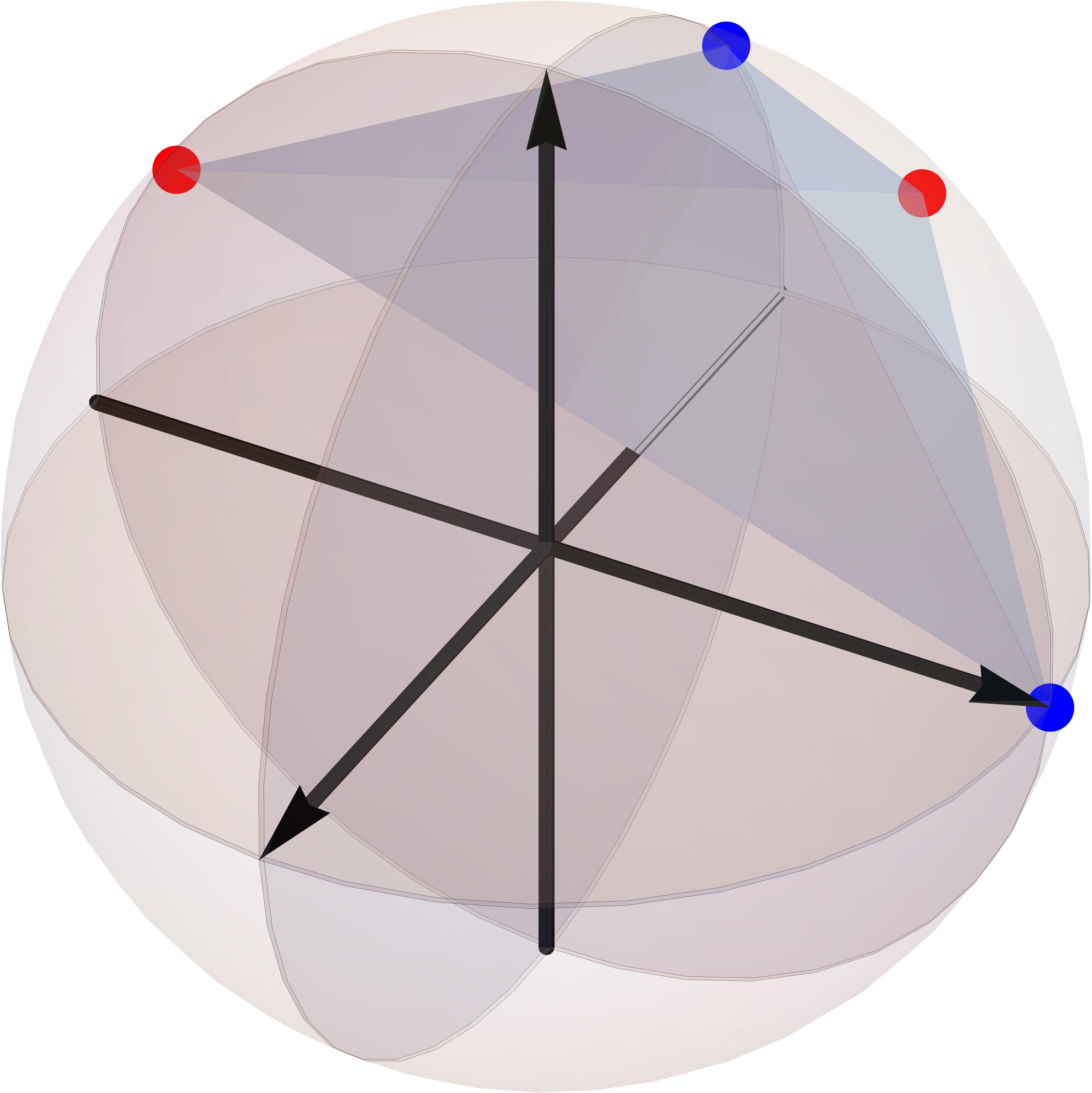}
		\caption{With discord.}
		\label{fig:withdiscord}
	\end{subfigure}
	\caption{Bloch representation of two local operators for system $A$ before (blue dots) and after (red dots) a local unitary rotation, for the case of a two-qubit state $\rho_{AB}$. (a) In the case of no discord, the blue dots correspond to the representation of two orthogonal pure states; red dots and blue dots are necessarily coplanar, independently of the unitary transformation, and hence do not span the entire three-dimensional (Bloch) space: channel tomography is not possible. (b) In the case with discord, the blue dots represent the (rescaled) Bloch component of two orthonormal (with respect to the Hilbert-Schmidt inner product) operators $A_1$ and $A_2$ that enter the OSD decomposition of $\rho_{AB}$ not trivially, and that do not commute; there is a unitary such that red dots and blue dots are not coplanar, and hence span the entire three-dimensional space: channel tomography is possible.}
\end{figure}

On the other hand, assume that $\rab$ has non-zero discord on $A$. This implies that there are some correlations, that is, that $\OSR(\rab)\geq 2$. Without loss of generality, we can assume that the OSD of $\rho$ must nontrivially contain $A_1$ and $A_2$ that do not commute, since, if all the non-trivial $A_i$'s that appear in the OSD of $\rho_{AB}$ commuted pairwise, they would all commute, and there would not be any discord. From Lemma~\ref{lem:commuting}, the Bloch vectors of $A_1$ and $A_2$ are not collinear. This means that there is a rotation $R$ of such vectors such that the resulting four vectors identify affinely independent points which span $\mathbb{R}^3$ (see Fig.\ref{fig:withdiscord}). Via the homomorphism between $SO(2)$ and $SU(3)$, the rotation $R$ corresponds to unitary rotation $U$ such that $\{\rab,U_A\otimes\I_B\rab U_A^\dag\otimes\I_B\}$ is faithful.

\end{proof}

\subsection{Mixed probe-ancilla state: examples of efficient generation of faithful sets}

In this example we present a family of mixed states of two qudits which generate faithful sets with even only two local unitaries, one being the identity. In order to construct the example we will make use of the Weyl (or generalized Pauli) basis for the space of $d\times d$ linear operators, which is given by $X^kZ^l$ with $k,l=0,\ldots,d-1$, where $X=\sum_{p=0}^{d-1}\projj{p+1}{p}$, $Z=\sum_{q=0}^{d-1}\omega^q\proj{q}$ and $\omega=e^{2\pi i/d}$ is a root of unity. Both $X$ and $Z$ are unitary, so that $X^\dagger = X^{-1}$ (similarly for $Z$). Since $X^d=Z^d=\I$, the sets $\{X^k\}_k$ and $\{Z^l\}_l$ form cyclic groups under multiplication, and we can think that the exponent is taken modulus $d$. Let $F=\frac{1}{\sqrt{d}}\sum_{k,l=0}^{d-1}w^{kl}\projj{k}{l}$ be the discrete Fourier transform unitary. One has $F X F^\dagger = Z$, $F Z F^\dagger = X^\dagger$, and the braiding relation $ZX=\omega XZ$, from which one deduces $F X^k Z^l F^\dagger=Z^kX^{-l}=\omega^{-kl}X^{-l}Z^{k}$, and that the action of $F\cdot F^\dag$ on the basis elements $X^kZ^l$ induces closed and disjoint \emph{orbits} within $\{X^kZ^l\}_{k,l}$ (up to irrelevant phases), defined as $O(k,l)=\{F^nX^kZ^lF^{\dag n}\ |\ n=0,\ldots,3\}$. Specifically, one has
\begin{align}
F X^k Z^l F^{\dagger} &= \omega^{-kl}X^{-l} Z^k,\\
F^2 X^k Z^l F^{\dagger 2} &= X^{-k} Z^{-l},\\
F^3 X^k Z^l F^{\dagger 3} &= \omega^{-kl} X^l Z^{-k},
\end{align}
while $F^4=\I$ so that obviously $F^4X^k Z^l F^{\dagger 4} = X^k Z^l$.
Such orbits contains either one, two, or four distinct elements. The only orbits with a single element are the one including the identity, corresponding to $(k,l)=(0,0)$, for both $d$ even and odd, and the one corresponding to the element $X^{d/2}Z^{d/2}$ for $d$ even. No orbit can contain exactly three distinct elements, as this would require that one of such elements is invariant under $F\cdot F^\dagger$, which is a contradiction for an orbit that contains more than one element and that is known to close necessarily under four repeated actions of $F\cdot F^\dagger$.

Consider the set $\mathcal{O}=\{O(k,l)\}$ of orbits (up to irrelevant phases). We identify $O(k,l)=O(k',l')$ if $X^kZ^l$ is in the same orbit as $X^{k'}Z^{l'}$ (up to irrelevant phases).
Consider the set $\mathcal{P}^{(1)}=\{W^{(1)}(O)|O\in\mathcal{O}\}$ composed of one Weyl-operator representative $W^{(1)}(O)$ per orbit $O$ (the exact choice of representative is irrelevant in this case). It is clear that $\cup_{i=0}^3F^i\mathcal{P}^{(1)}F^{\dagger i} = \{X^kZ^l\}_{k,l}$ up to irrelevant phases, where we have have used the shorthand notation $\Lambda\{K_m\}$ for the image $\{\Lambda[K_m]\}$ of a set $\{K_m\}$ under the action of a map $\Lambda$. In particular, $\vspan(\cup_{i=0}^3F^i\mathcal{P}^{(1)}F^{\dagger i})=\vspan(\{X^kZ^l\}_{k,l})$. Furthermore, it is also clear that there is a choice of \emph{pairs} of representatives (the two representatives may be chosen to coincide, in the case of 1-element and 2-element orbits) per orbit, forming a set $\mathcal{P}^{(2)}=\{W^{(2)}_1(O),W^{(2)}_2(O)|O\in\mathcal{O}\}$, such that $\cup_{i=0}^1F^i\mathcal{P}^{(2)}F^{\dagger i} = \{X^kZ^l\}_{k,l}$ (up to irrelevant phases).

Let us be more concrete, providing a specific choice for the set $\mathcal{P}^{(1)}$. We consider two cases, according to the parity of $d$. First, let $d=2m+1$ be odd. It is easy to verify that (up to irrelevant phases) we can pick
\begin{equation*}
\mathcal{P}^{(1)} = \{X^kZ^l\ |\ k=0,\ldots,m;l=1,\ldots,m\}\cup\{\I\}.
\end{equation*}
If $d=2m$ is instead even, then we can choose (again, up to irrelevant phases)
\begin{align*}
\mathcal{P}^{(1)}=\{X^kZ^l\ |\ k=0,\ldots,m-1;&l=1,\ldots,m\}\\
&\cup\{\I,X^mZ^m\}.
\end{align*}
Notice the the above sets $\mathcal{P}^{(1)}$ contain $(d^2+3)/4$ and $(d^2+8)/4$ elements, for odd and even dimension respectively, thus scaling approximately as $d^2/4$. 

One can similarly construct a set $\mathcal{P}^{(2)}\subset \{X^kZ^l\}_{k,l}$, with approximately $\approx d^2/2$ elements, such that $\mathcal{P}^{(2)}\cup F\mathcal{P}^{(2)}F^\dagger$ spans the entire operator space. For example, one can take $\mathcal{P}^{(2)}=\mathcal{P}^{(1)}\cup F\mathcal{P}^{(1)}F^\dagger$, where elements that are identical up to a phase factor are equated.

The issue we still have to face is how to use the facts above to construct, e.g., a state $\sigma_{AB}\in \cD(\mathbb{C}^d\otimes \mathbb{C}^d)$ with $\OSD(\sigma_{AB})\approx d^2/2$ such that $\{\sigma_{AB},U_A\otimes \I_A \sigma_{AB} U_A^\dagger\otimes \I_A\}$ is a faithful set for channel tomography on qudit $A$. This is easily done by considering the hermitian operators
\begin{align}
H_{k,l}&=\frac{1}{\sqrt{2}}\left(X^kZ^l+(X^kZ^l)^\dagger\right) \\
	&= \frac{1}{\sqrt{2}}\left(X^kZ^l+\omega^{kl}F^2(X^kZ^l)F^{\dagger 2}\right),\\ 
J_{k,l}&=\frac{1}{i\sqrt{2}}\left(X^kZ^l-(X^kZ^l)^\dag\right)\\
	&= \frac{1}{i\sqrt{2}}\left(X^kZ^l-\omega^{kl}F^2(X^kZ^l)F^{\dagger 2}\right)
\end{align}
where we have used the relation $F^2 X^k Z^l F^{\dagger 2} = X^{-k} Z^{-l}$ and the braiding relation of $X$ and $Z$.
The operators $H_{k,l}$ and $J_{k,l}$ obviously span the same operator subspace as $X^kZ^l$ and $F^2(X^kZ^l)F^{\dag 2}$. 
Furthermore, we will use the fact that, given any Hermitian operator $H$ in finite dimensions, there is $\epsilon>0$ small enough (more precisely, it is enough that $\epsilon |\lambda_{-}(H)|\leq 1$, with $\lambda_{-}(H)$ the most negative eigenvalue of $H$) such that $\I+\epsilon H$ is positive semidefinite. Thus, for example, given our choice of $\mathcal{P}^{(1)}$ above, it is clear that, for, say, $d$ odd (the even case is handled similarly), we can take
\[
\begin{split}
\sigma_{AB}
&\propto\frac{\I_{A}}{d}\otimes\frac{\I_{B}}{d} \\
&\quad+ \epsilon\sum_{k=0}^{(d-1)/2} \sum_{l=1}^{ (d-1)/2 }\left( H_{k,l}\otimes H_{k,l}+J_{k,l}\otimes J_{k,l}\right).
\end{split}
\]
with $\epsilon>0$ small enough. Then, by construction, $\left\{\sigma_{AB}\,,\,(F_A\otimes \I_B) \sigma_{AB}  (F_A\otimes \I_B)^\dagger\right\}$ is a faithful set. Notice that $\sigma_{AB}$ has OSR less or equal to $2\frac{d-1}{2}\left(\frac{d-1}{2}+1\right)+1 = \frac{d^2+1}{2}$.

\section{Conclusions}

We have introduced and analyzed some properties of a framework for process tomography assisted by correlations. Our framework interpolates between standard process tomography and ancilla-assisted process tomography, and it is based on applying local transformations on the input probe---part of an probe-ancilla bipartite system---before the probe undergoes the process to be reconstructed. In particular, we focused on determining how the correlation properties of the starting probe-ancilla state $\rho_{AB}$ affect such a number. We proved that essentially all correlations can be helpful, in the sense of reducing such a number from $d^2$ for standard process tomography to roughly $d^2/\OSR(\rho_{AB})$, with $\OSR(\rho_{AB})$ the operator Schmidt rank of $\rho_{AB}$, which is necessarily optimal. We proved that this holds true in the pure-input case even if the local transformations are restricted to be unitary. In the mixed-state case, we pointed out the role of discord in the case of qubit probes and unitary local transformations, and gave ``extreme'' examples where just one additional initial local unitary rotation suffices for process tomography, even if the initial state has operator Schmidt rank approximately $d^2/2$. It would be interesting to fully understand the mixed-state case for unitary local rotations, which appears to be related to studying and applying the adjoint representation of the unitary group, and will be investigated in future work.

\begin{acknowledgments}
	The authors thank J. Watrous for discussions. MP acknowledges support from European Union's Horizon 2020 Research and Innovation Programme under the Marie Sk{\l}odowska-Curie Action OPERACQC (Grant Agreement No. 661338), and from the Foundational Questions Institute under the Physics of the Observer Programme (Grant No. FQXi-RFP-1601). 
\end{acknowledgments}

%
%
%
%
%
%
%


\begin{thebibliography}{19}%
	\makeatletter
	\providecommand \@ifxundefined [1]{%
		\@ifx{#1\undefined}
	}%
	\providecommand \@ifnum [1]{%
		\ifnum #1\expandafter \@firstoftwo
		\else \expandafter \@secondoftwo
		\fi
	}%
	\providecommand \@ifx [1]{%
		\ifx #1\expandafter \@firstoftwo
		\else \expandafter \@secondoftwo
		\fi
	}%
	\providecommand \natexlab [1]{#1}%
	\providecommand \enquote  [1]{``#1''}%
	\providecommand \bibnamefont  [1]{#1}%
	\providecommand \bibfnamefont [1]{#1}%
	\providecommand \citenamefont [1]{#1}%
	\providecommand \href@noop [0]{\@secondoftwo}%
	\providecommand \href [0]{\begingroup \@sanitize@url \@href}%
	\providecommand \@href[1]{\@@startlink{#1}\@@href}%
	\providecommand \@@href[1]{\endgroup#1\@@endlink}%
	\providecommand \@sanitize@url [0]{\catcode `\\12\catcode `\$12\catcode
		`\&12\catcode `\#12\catcode `\^12\catcode `\_12\catcode `\%12\relax}%
	\providecommand \@@startlink[1]{}%
	\providecommand \@@endlink[0]{}%
	\providecommand \url  [0]{\begingroup\@sanitize@url \@url }%
	\providecommand \@url [1]{\endgroup\@href {#1}{\urlprefix }}%
	\providecommand \urlprefix  [0]{URL }%
	\providecommand \Eprint [0]{\href }%
	\providecommand \doibase [0]{http://dx.doi.org/}%
	\providecommand \selectlanguage [0]{\@gobble}%
	\providecommand \bibinfo  [0]{\@secondoftwo}%
	\providecommand \bibfield  [0]{\@secondoftwo}%
	\providecommand \translation [1]{[#1]}%
	\providecommand \BibitemOpen [0]{}%
	\providecommand \bibitemStop [0]{}%
	\providecommand \bibitemNoStop [0]{.\EOS\space}%
	\providecommand \EOS [0]{\spacefactor3000\relax}%
	\providecommand \BibitemShut  [1]{\csname bibitem#1\endcsname}%
	\let\auto@bib@innerbib\@empty
	\bibitem [{\citenamefont {Nielsen}\ and\ \citenamefont
		{Chuang}(2010)}]{nielsen2010quantum}%
	\BibitemOpen
	\bibfield  {author} {\bibinfo {author} {\bibfnamefont {M.~A.}\ \bibnamefont
			{Nielsen}}\ and\ \bibinfo {author} {\bibfnamefont {I.~L.}\ \bibnamefont
			{Chuang}},\ }\href@noop {} {\emph {\bibinfo {title} {Quantum Computation and
				Quantum Information}}}\ (\bibinfo  {publisher} {Cambridge University Press},\
	\bibinfo {year} {2010})\BibitemShut {NoStop}%
	\bibitem [{\citenamefont {Chuang}\ and\ \citenamefont
		{Nielsen}(1997)}]{chuang1997prescription}%
	\BibitemOpen
	\bibfield  {author} {\bibinfo {author} {\bibfnamefont {I.~L.}\ \bibnamefont
			{Chuang}}\ and\ \bibinfo {author} {\bibfnamefont {M.~A.}\ \bibnamefont
			{Nielsen}},\ }\href@noop {} {\bibfield  {journal} {\bibinfo  {journal}
			{Journal of Modern Optics}\ }\textbf {\bibinfo {volume} {44}},\ \bibinfo
		{pages} {2455} (\bibinfo {year} {1997})}\BibitemShut {NoStop}%
	\bibitem [{\citenamefont {Poyatos}\ \emph {et~al.}(1997)\citenamefont
		{Poyatos}, \citenamefont {Cirac},\ and\ \citenamefont
		{Zoller}}]{poyatos1997complete}%
	\BibitemOpen
	\bibfield  {author} {\bibinfo {author} {\bibfnamefont {J.}~\bibnamefont
			{Poyatos}}, \bibinfo {author} {\bibfnamefont {J.~I.}\ \bibnamefont {Cirac}},
		\ and\ \bibinfo {author} {\bibfnamefont {P.}~\bibnamefont {Zoller}},\
	}\href@noop {} {\bibfield  {journal} {\bibinfo  {journal} {Physical Review
				Letters}\ }\textbf {\bibinfo {volume} {78}},\ \bibinfo {pages} {390}
		(\bibinfo {year} {1997})}\BibitemShut {NoStop}%
	\bibitem [{\citenamefont {D'Ariano}\ \emph {et~al.}(2003)\citenamefont
		{D'Ariano}, \citenamefont {Paris},\ and\ \citenamefont
		{Sacchi}}]{d2003quantum}%
	\BibitemOpen
	\bibfield  {author} {\bibinfo {author} {\bibfnamefont {G.~M.}\ \bibnamefont
			{D'Ariano}}, \bibinfo {author} {\bibfnamefont {M.~G.}\ \bibnamefont {Paris}},
		\ and\ \bibinfo {author} {\bibfnamefont {M.~F.}\ \bibnamefont {Sacchi}},\
	}\href@noop {} {\bibfield  {journal} {\bibinfo  {journal} {Advances in
				Imaging and Electron Physics}\ }\textbf {\bibinfo {volume} {128}},\ \bibinfo
		{pages} {206} (\bibinfo {year} {2003})}\BibitemShut {NoStop}%
	\bibitem [{\citenamefont {Bisio}\ \emph {et~al.}(2009)\citenamefont {Bisio},
		\citenamefont {Chiribella}, \citenamefont {D'Ariano}, \citenamefont
		{Facchini},\ and\ \citenamefont {Perinotti}}]{bisio2009optimal}%
	\BibitemOpen
	\bibfield  {author} {\bibinfo {author} {\bibfnamefont {A.}~\bibnamefont
			{Bisio}}, \bibinfo {author} {\bibfnamefont {G.}~\bibnamefont {Chiribella}},
		\bibinfo {author} {\bibfnamefont {G.}~\bibnamefont {D'Ariano}}, \bibinfo
		{author} {\bibfnamefont {S.}~\bibnamefont {Facchini}}, \ and\ \bibinfo
		{author} {\bibfnamefont {P.}~\bibnamefont {Perinotti}},\ }\href@noop {}
	{\bibfield  {journal} {\bibinfo  {journal} {Physical review letters}\
		}\textbf {\bibinfo {volume} {102}},\ \bibinfo {pages} {010404} (\bibinfo
		{year} {2009})}\BibitemShut {NoStop}%
	\bibitem [{\citenamefont {D'Ariano}\ and\ \citenamefont
		{Presti}(2001)}]{d2001quantum}%
	\BibitemOpen
	\bibfield  {author} {\bibinfo {author} {\bibfnamefont {G.}~\bibnamefont
			{D'Ariano}}\ and\ \bibinfo {author} {\bibfnamefont {P.~L.}\ \bibnamefont
			{Presti}},\ }\href@noop {} {\bibfield  {journal} {\bibinfo  {journal}
			{Physical review letters}\ }\textbf {\bibinfo {volume} {86}},\ \bibinfo
		{pages} {4195} (\bibinfo {year} {2001})}\BibitemShut {NoStop}%
	\bibitem [{\citenamefont {Altepeter}\ \emph {et~al.}(2003)\citenamefont
		{Altepeter}, \citenamefont {Branning}, \citenamefont {Jeffrey}, \citenamefont
		{Wei}, \citenamefont {Kwiat}, \citenamefont {Thew}, \citenamefont {O'Brien},
		\citenamefont {Nielsen},\ and\ \citenamefont {White}}]{altepeter2003ancilla}%
	\BibitemOpen
	\bibfield  {author} {\bibinfo {author} {\bibfnamefont {J.~B.}\ \bibnamefont
			{Altepeter}}, \bibinfo {author} {\bibfnamefont {D.}~\bibnamefont {Branning}},
		\bibinfo {author} {\bibfnamefont {E.}~\bibnamefont {Jeffrey}}, \bibinfo
		{author} {\bibfnamefont {T.}~\bibnamefont {Wei}}, \bibinfo {author}
		{\bibfnamefont {P.~G.}\ \bibnamefont {Kwiat}}, \bibinfo {author}
		{\bibfnamefont {R.~T.}\ \bibnamefont {Thew}}, \bibinfo {author}
		{\bibfnamefont {J.~L.}\ \bibnamefont {O'Brien}}, \bibinfo {author}
		{\bibfnamefont {M.~A.}\ \bibnamefont {Nielsen}}, \ and\ \bibinfo {author}
		{\bibfnamefont {A.~G.}\ \bibnamefont {White}},\ }\href@noop {} {\bibfield
		{journal} {\bibinfo  {journal} {Physical Review Letters}\ }\textbf {\bibinfo
			{volume} {90}},\ \bibinfo {pages} {193601} (\bibinfo {year}
		{2003})}\BibitemShut {NoStop}%
	\bibitem [{\citenamefont {D'Ariano}\ and\ \citenamefont
		{Presti}(2003)}]{d2003imprinting}%
	\BibitemOpen
	\bibfield  {author} {\bibinfo {author} {\bibfnamefont {G.~M.}\ \bibnamefont
			{D'Ariano}}\ and\ \bibinfo {author} {\bibfnamefont {P.~L.}\ \bibnamefont
			{Presti}},\ }\href@noop {} {\bibfield  {journal} {\bibinfo  {journal}
			{Physical review letters}\ }\textbf {\bibinfo {volume} {91}},\ \bibinfo
		{pages} {047902} (\bibinfo {year} {2003})}\BibitemShut {NoStop}%
	\bibitem [{\citenamefont {Choi}(1975)}]{choi1975completely}%
	\BibitemOpen
	\bibfield  {author} {\bibinfo {author} {\bibfnamefont {M.-D.}\ \bibnamefont
			{Choi}},\ }\href@noop {} {\bibfield  {journal} {\bibinfo  {journal} {Linear
				algebra and its applications}\ }\textbf {\bibinfo {volume} {10}},\ \bibinfo
		{pages} {285} (\bibinfo {year} {1975})}\BibitemShut {NoStop}%
	\bibitem [{\citenamefont {Jamio{\l}kowski}(1972)}]{jamiolkowski1972linear}%
	\BibitemOpen
	\bibfield  {author} {\bibinfo {author} {\bibfnamefont {A.}~\bibnamefont
			{Jamio{\l}kowski}},\ }\href@noop {} {\bibfield  {journal} {\bibinfo
			{journal} {Reports on Mathematical Physics}\ }\textbf {\bibinfo {volume}
			{3}},\ \bibinfo {pages} {275} (\bibinfo {year} {1972})}\BibitemShut {NoStop}%
	\bibitem [{\citenamefont {Caiaffa}\ and\ \citenamefont
		{Piani}(2018)}]{caiaffa2018}%
	\BibitemOpen
	\bibfield  {author} {\bibinfo {author} {\bibfnamefont {M.}~\bibnamefont
			{Caiaffa}}\ and\ \bibinfo {author} {\bibfnamefont {M.}~\bibnamefont
			{Piani}},\ }\href {\doibase 10.1103/PhysRevA.97.032334} {\bibfield  {journal}
		{\bibinfo  {journal} {Phys. Rev. A}\ }\textbf {\bibinfo {volume} {97}},\
		\bibinfo {pages} {032334} (\bibinfo {year} {2018})}\BibitemShut {NoStop}%
	\bibitem [{\citenamefont {Horodecki}\ \emph {et~al.}(2009)\citenamefont
		{Horodecki}, \citenamefont {Horodecki}, \citenamefont {Horodecki},\ and\
		\citenamefont {Horodecki}}]{revent}%
	\BibitemOpen
	\bibfield  {author} {\bibinfo {author} {\bibfnamefont {R.}~\bibnamefont
			{Horodecki}}, \bibinfo {author} {\bibfnamefont {P.}~\bibnamefont
			{Horodecki}}, \bibinfo {author} {\bibfnamefont {M.}~\bibnamefont
			{Horodecki}}, \ and\ \bibinfo {author} {\bibfnamefont {K.}~\bibnamefont
			{Horodecki}},\ }\href {\doibase 10.1103/RevModPhys.81.865} {\bibfield
		{journal} {\bibinfo  {journal} {Rev. Mod. Phys.}\ }\textbf {\bibinfo {volume}
			{81}},\ \bibinfo {pages} {865} (\bibinfo {year} {2009})}\BibitemShut
	{NoStop}%
	\bibitem [{\citenamefont {Aniello}\ and\ \citenamefont
		{Lupo}(2009)}]{aniello2009relation}%
	\BibitemOpen
	\bibfield  {author} {\bibinfo {author} {\bibfnamefont {P.}~\bibnamefont
			{Aniello}}\ and\ \bibinfo {author} {\bibfnamefont {C.}~\bibnamefont {Lupo}},\
	}\href@noop {} {\bibfield  {journal} {\bibinfo  {journal} {Open Systems \&
				Information Dynamics}\ }\textbf {\bibinfo {volume} {16}},\ \bibinfo {pages}
		{127} (\bibinfo {year} {2009})}\BibitemShut {NoStop}%
	\bibitem [{\citenamefont {Lupo}\ \emph {et~al.}(2008)\citenamefont {Lupo},
		\citenamefont {Aniello},\ and\ \citenamefont
		{Scardicchio}}]{lupo2008bipartite}%
	\BibitemOpen
	\bibfield  {author} {\bibinfo {author} {\bibfnamefont {C.}~\bibnamefont
			{Lupo}}, \bibinfo {author} {\bibfnamefont {P.}~\bibnamefont {Aniello}}, \
		and\ \bibinfo {author} {\bibfnamefont {A.}~\bibnamefont {Scardicchio}},\
	}\href@noop {} {\bibfield  {journal} {\bibinfo  {journal} {Journal of Physics
				A: Mathematical and Theoretical}\ }\textbf {\bibinfo {volume} {41}},\
		\bibinfo {pages} {415301} (\bibinfo {year} {2008})}\BibitemShut {NoStop}%
	\bibitem [{\citenamefont {Duffin}\ and\ \citenamefont
		{Schaeffer}(1952)}]{duffin1952class}%
	\BibitemOpen
	\bibfield  {author} {\bibinfo {author} {\bibfnamefont {R.~J.}\ \bibnamefont
			{Duffin}}\ and\ \bibinfo {author} {\bibfnamefont {A.~C.}\ \bibnamefont
			{Schaeffer}},\ }\href@noop {} {\bibfield  {journal} {\bibinfo  {journal}
			{Transactions of the American Mathematical Society}\ }\textbf {\bibinfo
			{volume} {72}},\ \bibinfo {pages} {341} (\bibinfo {year} {1952})}\BibitemShut
	{NoStop}%
	\bibitem [{\citenamefont {Werner}(1989)}]{werner1989quantum}%
	\BibitemOpen
	\bibfield  {author} {\bibinfo {author} {\bibfnamefont {R.~F.}\ \bibnamefont
			{Werner}},\ }\href@noop {} {\bibfield  {journal} {\bibinfo  {journal}
			{Physical Review A}\ }\textbf {\bibinfo {volume} {40}},\ \bibinfo {pages}
		{4277} (\bibinfo {year} {1989})}\BibitemShut {NoStop}%
	\bibitem [{\citenamefont {Henderson}\ and\ \citenamefont
		{Vedral}(2001)}]{henderson2001classical}%
	\BibitemOpen
	\bibfield  {author} {\bibinfo {author} {\bibfnamefont {L.}~\bibnamefont
			{Henderson}}\ and\ \bibinfo {author} {\bibfnamefont {V.}~\bibnamefont
			{Vedral}},\ }\href@noop {} {\bibfield  {journal} {\bibinfo  {journal} {J.
				Phys. A: Math. Gen.}\ }\textbf {\bibinfo {volume} {34}},\ \bibinfo {pages}
		{6899} (\bibinfo {year} {2001})}\BibitemShut {NoStop}%
	\bibitem [{\citenamefont {Ollivier}\ and\ \citenamefont
		{Zurek}(2001)}]{ollivier2001quantum}%
	\BibitemOpen
	\bibfield  {author} {\bibinfo {author} {\bibfnamefont {H.}~\bibnamefont
			{Ollivier}}\ and\ \bibinfo {author} {\bibfnamefont {W.~H.}\ \bibnamefont
			{Zurek}},\ }\href@noop {} {\bibfield  {journal} {\bibinfo  {journal} {Phys.
				Rev. Lett.}\ }\textbf {\bibinfo {volume} {88}},\ \bibinfo {pages} {017901}
		(\bibinfo {year} {2001})}\BibitemShut {NoStop}%
	\bibitem [{\citenamefont {Modi}\ \emph {et~al.}(2012)\citenamefont {Modi},
		\citenamefont {Brodutch}, \citenamefont {Cable}, \citenamefont {Paterek},\
		and\ \citenamefont {Vedral}}]{RevModPhys.84.1655}%
	\BibitemOpen
	\bibfield  {author} {\bibinfo {author} {\bibfnamefont {K.}~\bibnamefont
			{Modi}}, \bibinfo {author} {\bibfnamefont {A.}~\bibnamefont {Brodutch}},
		\bibinfo {author} {\bibfnamefont {H.}~\bibnamefont {Cable}}, \bibinfo
		{author} {\bibfnamefont {T.}~\bibnamefont {Paterek}}, \ and\ \bibinfo
		{author} {\bibfnamefont {V.}~\bibnamefont {Vedral}},\ }\href {\doibase
		10.1103/RevModPhys.84.1655} {\bibfield  {journal} {\bibinfo  {journal} {Rev.
				Mod. Phys.}\ }\textbf {\bibinfo {volume} {84}},\ \bibinfo {pages} {1655}
		(\bibinfo {year} {2012})}\BibitemShut {NoStop}%
\end{thebibliography}
%

%
%
%
%
%
%
%

\end{document}